\documentclass[11pt]{amsart}

\usepackage{amssymb,latexsym}
\usepackage{graphicx,epsfig,color}

\newtheorem{thm}{Theorem}[section]

\newtheorem{prop}[thm]{Proposition}

\newtheorem{lem}[thm]{Lemma}
\newtheorem{rem}[thm]{Remark}

\newtheorem{defn}[thm]{Definition}

\numberwithin{equation}{section}

\date{}

\begin{document}

\author[Tulkin H. Rasulov and Nargiza A. Tosheva]{Tulkin H. Rasulov and Nargiza A. Tosheva}
\title[Analytic description of the essential spectrum...]
{Analytic description of the essential spectrum of a family of $3 \times 3$ operator
matrices} \maketitle

\begin{center}
{\small Faculty of Physics and Mathematics, Bukhara State University\\
M. Ikbol str. 11, 200100 Bukhara, Uzbekistan\\
E-mail: rth@mail.ru, nargiza$_{-}$n@mail.ru}
\end{center}

\begin{abstract}
We consider the family of $3 \times 3$ operator matrices
$H(K),$ $K \in {\Bbb T}^{\rm d}:=(-\pi; \pi]^{\rm d}$ arising in the spectral analysis of the energy operator of the spin-boson model of radioactive decay with two bosons on the torus ${\Bbb T}^{\rm d}.$ We obtain an analogue of the Faddeev
equation for the eigenfunctions of $H(K)$. An analytic description of the
essential spectrum of $H(K)$ is established. Further,
it is shown that the essential spectrum of $H(K)$
consists the union of at most three bounded closed intervals.
\end{abstract}

\medskip {AMS subject Classifications:} Primary 81Q10; Secondary
35P20, 47N50.

\textbf{Key words and phrases:} family of operator matrices, generalized Friedrichs model, bosonic Fock
space, annihilation and creation operators, channel operator, decomposable operator, fiber operators, the Faddeev equation,
essential spectrum, Weyl criterion.

\section{Introduction}

In statistical physics \cite{MS}, solid-state physics \cite{Mog}
and the theory of quantum fields \cite{Frid}, one considers systems,
where the number of quasi-particles is not fixed.
Their number can be unbounded as in the case of full spin-boson models (infinite operator matrix)
\cite{HS95} or bounded as in the case of "truncated" $ $ spin-boson models (finite operator matrix)
\cite{HS94, MS, MNR2015, OI2018, Ras2016, ZhM95}.
Often, the number of particles can be arbitrary large as in cases
involving photons, in other cases, such as scattering of spin waves
on defects, scattering massive particles and chemical reactions,
there are only participants at any given time, though their number
can be change.

Recall that the study of systems describing $N$ particles in
interaction, without conservation of the number of particles is
reduced to the investigation of the spectral properties of
self-adjoint operators, acting in the {\it cut subspace} ${\mathcal
H}^{(N)}$ of Fock space, consisting of $n\leq N$ particles
\cite{Frid,MS,Mog,SSZ}.

The essential spectrum of the Hamiltonians (matrix operators) in the Fock space
are the most actively studied objects in operator theory. One of the important
problems in the spectral analysis of these operators is to describe the location
of the essential spectrum. One of the well-known methods for the investigation of the location of
essential spectra of operator matrices are Weyl criterion and the Hunziker-van Winter-Zhislin (HWZ) theorem.
Using these methods in many works the essential spectrum of the $3 \times 3$ and
$4 \times 4$ lattice operator matrices are studied, see e.g., \cite{LR03, MR2015, Ras2010, RMH}.
In particular, in \cite{RMH} it is described the essential spectrum of
$4 \times 4$ operator matrix by the spectrum of the corresponding channel operators and
proved the HWZ theorem.
In \cite{SSZ} geometric and commutator
techniques have been developed in order to find the location of the
spectrum and to prove absence of singular continuous spectrum for
Hamiltonians without conservation of the particle number.

In the present paper we consider the family of $3 \times 3$
operator matrices $H(K),$ $K \in {\Bbb T}^{\rm d}$ associated with
the lattice systems describing two identical bosons and one
particle, another nature in interactions, without conservation of
the number of particles. This operator acts in the direct sum of
zero-, one- and two-particle subspaces of the bosonic Fock space
and it is a lattice analogue of the spin-boson Hamiltonian
\cite{MS}. For the study of location of the essential spectrum of $H(K)$, first
we introduce the notion of {\it channel operator} $H_{\rm ch}(K)$ corresponding to $H(K)$.
Using the theorem on spectrum of decomposable operators we describe the spectrum of $H_{\rm ch}(K)$
via the spectrum of a family of generalized Friedrichs models $h(K,k),$ $K,k \in {\Bbb T}^{\rm d}.$
Then we prove that the essential spectrum of $H(K)$ is coincide with the spectrum of $H_{\rm ch}(K)$ and
show that the set $\sigma_{\rm ess}(H(K))$ consists the union of at most 3 bounded closed intervals.
Further, we define the new so-called {\it two- and three-particle branches}
of $\sigma_{\rm ess}(H(K))$.

We point out that the operator $H(K)$ has been considered
before in \cite{ALR, ALR1} for a fixed $K \in {\Bbb T}^{\rm d}$ and studied its essential and discrete spectrum. In this paper we investigate the essential spectrum of $H(K)$ depending on $K\in {\Bbb T}^{\rm d}$.
It is remarkable that, the essential spectrum and the number of the eigenvalues of a slightly simpler version of $H(K)$
were studied in \cite{MR14} and the result about the location of the essential spectrum were announced without proof.

The present paper is organized as follows. Section 1 is
an introduction to the whole work. In Section 2, the operator
matrices $H(K),$ $K \in {\Bbb T}^{\rm d}$ are described as the
family of bounded self-adjoint operators in the direct sum of
zero-, one- and two-particle subspaces of the bosonic Fock space and
the main aims of the paper are stated. Family of generalized Friedrichs model
is introduced and its spectrum is established in Section 3.
In Section 4, we determine the channel operator $H_{\rm ch}(K)$ corresponding to $H(K)$ and
define its spectrum.
Next Section is devoted to the derivation
of the Faddeev equation for the eigenvectors of $H(K).$ In Section 6, we investigate the essential spectrum of $H(K)$
and its new branches.

Throughout this paper, we use the following notation. If $A$ is a linear bounded self-adjoint operator
from Hilbert space to another, then
$\sigma(A)$ denotes its spectrum, $\sigma_{\rm
ess}(A)$ its essential spectrum and $\sigma_{\rm disc}(A)$ its discrete spectrum.

\section{Family of $3 \times 3$ operator matrices and its relation with spin-boson models}

Let us introduce some notations used in this work.
Let ${\Bbb T}^{\rm d}$ be the ${\rm d}$-dimensional torus, the cube
$(-\pi,\pi]^{\rm d}$ with appropriately identified sides equipped with
its Haar measure.
Let ${\mathcal H}_0:={\Bbb C}$ be the field of complex numbers (channel 1), ${\mathcal H}_1:=L^2({\Bbb
T}^{\rm d})$ be the Hilbert space of square integrable (complex) functions
defined on ${\Bbb T}^{\rm d}$ (channel 2) and ${\mathcal H}_2:=L^2_{\rm sym}(({\Bbb T}^{\rm d})^2)$ stands for the subspace
of $L^2(({\Bbb T}^{\rm d})^2)$ consisting of symmetric functions (with respect to the two variables) (channel 3).
Denote by ${\mathcal H}$ the direct
sum of spaces ${\mathcal H}_0,$ ${\mathcal H}_1$ and ${\mathcal H}_2,$ that is,
${\mathcal H}={\mathcal H}_0 \oplus {\mathcal H}_1 \oplus
{\mathcal H}_2.$ The spaces ${\mathcal H}_0,$ ${\mathcal H}_1$ and
${\mathcal H}_2$ are called zero-, one- and two-particle subspaces
of a bosonic Fock space ${\mathcal F}_{\rm s}(L^2({\Bbb T}^{\rm d}))$
over $L^2({\Bbb T}^{\rm d}),$ respectively, and the Hilbert space ${\mathcal H}$
is called the {\it truncated} Fock space or the {\it three-particle cut subspace} of the Fock space.
We write elements $f$ of the space $\mathcal H$ in the form, $f=(f_0, f_1, f_2),$ $f_i \in {\mathcal H}_i,$
$i=0,1,2$ and for any two elements $f=(f_0,f_1,f_2), g=(g_0,g_1,g_2) \in {\mathcal H}$
their scalar product is defined by
$$
(f,g):=f_0 \overline{g_0}+\int_{{\Bbb T}^{\rm d}} f_1(t) \overline{g_1(t)} dt+
\int_{({\Bbb T}^{\rm d})^2} f_2(s,t) \overline{g_2(s,t)} dsdt.
$$

In the present paper we study the essential spectrum of the family of $3 \times 3$ tridiagonal operator matrices
\begin{equation}\label{main operator}
H(K):=\left( \begin{array}{ccc}
H_{00}(K) & H_{01} & 0\\
H_{01}^* & H_{11}(K) & H_{12}\\
0 & H_{12}^* & H_{22}(K)\\
\end{array}
\right), \quad K \in {\Bbb T}^{\rm d}
\end{equation}
in the Hilbert space ${\mathcal H}.$
The matrix entries $H_{ii}(K): {\mathcal H}_i \to {\mathcal H}_i,$ $i=0,1,2$ and $H_{ij}:{\mathcal H}_j \to {\mathcal H}_i,$
$i<j,$ $i,j=0,1,2$ are given by
$$
H_{00}(K)f_0=w_0(K)f_0, \quad H_{01}f_1=\int_{{\Bbb T}^{\rm d}}
v_0(t)f_1(t)dt,
$$
$$
(H_{11}(K)f_1)(p)=w_1(K; p)f_1(p), \quad
(H_{12}f_2)(p)= \int_{{\Bbb T}^{\rm d}} v_1(t) f_2(p,t)dt,
$$
$$
(H_{22}(K)f_2)(p,q)=w_2(K;p,q)f_2(p,q),\quad f_i \in {\mathcal H}_i,\quad i=0,1,2.
$$

Throughout the paper we assume that the parameter functions $w_0(\cdot)$, $v_i(\cdot),$ $i=0,1$; $w_1(\cdot;
\cdot)$ and $w_2(\cdot; \cdot, \cdot)$ are real-valued continuous functions on
${\Bbb T}^{\rm d};$ $({\Bbb T}^{\rm d})^2$ and $({\Bbb T}^{\rm d})^3$, respectively.
In addition, for any fixed $K \in {\Bbb T}^{\rm d}$ the function $w_2(K; \cdot, \cdot)$ is a symmetric,
that is, the equality $w_2(K; p, q)=w_2(K; q, p)$ holds for any $p,q \in {\Bbb T}^{\rm d}$.

Under these assumptions the operator $H(K)$ is bounded
and self-adjoint.

We remark that the operators $H_{01}$ and $H_{12}$ resp.
$H_{01}^*$ and $H_{12}^*$ are called annihilation resp. creation
operators, respectively. A trivial verification shows that
\begin{align*}
& H_{01}^*: {\mathcal H}_0 \to {\mathcal H}_1, \quad (H_{01}^*f_0)(p)=v_0(p)f_0, \quad f_0 \in {\mathcal H}_0; \\
& H_{12}^*: {\mathcal H}_1 \to {\mathcal H}_2, \quad (H_{12}^*f_1)(p,q)=\frac{1}{2}(v_1(p)f_1(q)+v_1(q)f_1(p)), \quad f_1 \in {\mathcal H}_1.
\end{align*}
These operators have widespread applications in quantum mechanics,
notably in the study of quantum harmonic oscillators and many-particle systems \cite{Feyn}.
An annihilation operator lowers the number of particles in a given state by one.
A creation operator increases the number of particles in a given state by one, and it is the adjoint of the annihilation operator.
In many subfields of physics and chemistry, the use of these operators instead of wavefunctions is known as second quantization.
In this paper we consider the case, where
the number of annihilations and creations of the particles of the
considering system is equal to 1. It means that $H_{ij} \equiv 0$
for all $|i-j|>1.$

The family of operator matrices of this form play a key role for the study of the energy operator of the spin-boson model with two bosons on the torus. In fact, the latter
is a $6 \times 6$ operator matrix which is unitary equivalent to a $2 \times 2$ block diagonal operator with two copies of a particular case of $H(K)$ on the diagonal, see \cite{MNR2015}. Consequently, the essential spectrum and finiteness of discrete eigenvalues of the spin-boson Hamiltonian are determined by the corresponding spectral information on the operator matrix $H(K)$ in \eqref{main operator}.

Independently of whether the underlying domain is a torus or the whole space or the whole space ${\Bbb R}^{\rm d}$, the full spin-boson Hamiltonian is an infinite operator matrix in Fock space for which rigorous results are very hard to obtain. One line of attack is to consider the compression to the truncated Fock space with a finite $N$ of bosons, and in fact most of the existing literature concentrates on the case $N \leq 2.$ For the case of ${\Bbb R}^{\rm d}$ there some exceptions, see e.g. \cite{HS94, HS95} for arbitrary finite $N$ and \cite{ZhM95} for $N=3$, where a rigorous scattering theory was developed for small coupling constants.

For the case when the underlying domain is a torus, the spectral properties $H(K)$ for a fixed $K$ were investigated in \cite{ALR, ALR1, LR03, MNR2015, RMH}, see also the references therein. The results obtained in this paper for all $K$ will play important role when we study the problem of finding subset
$\Lambda \subset {\Bbb T}^{\rm d}$ such that the operator matrices $H(K)$ has a finitely or infinitely many eigenvalues for all $K \in \Lambda.$

It is well-known that the three-particle discrete Schr\"{o}dinger operator $\widehat{H}$ in the momentum representation is the bounded self-adjoint operator on the Hilbert space $L_2(({\Bbb T}^{\rm d})^3).$ Introducing the total quasimomentum $K \in {\Bbb T}^{\rm d}$ of the system, it is easy to see that the operator $\widehat{H}$ can be decomposed into the direct integral of the family $\{\widehat{H}(K),\,K \in {\Bbb T}^{\rm d}\}$
of self-adjoint operators \cite{ALM07, LM03}:
$$
\widehat{H}=\int_{{\Bbb T}^{\rm d}} \oplus \widehat{H}(K) dK,
$$
where the operator $\widehat{H}(K)$ acts on the Hilbert space $L_2(\Gamma_K)$ ($\Gamma_K \subset ({\Bbb T}^{\rm d})^2$ is some manifold).

Observe that $H(K)$ enjoys the main spectral properties of the three-particle discrete Schr\"{o}dinger operator $\widehat{H}(K)$ (see \cite{ALM07, LM03}), and the generalized Friedrichs model plays the role of the two-particle discrete Schr\"{o}dinger operator. For this reason the Hilbert space ${\mathcal H}$ is called the {\it three-particle cut subspace} of the Fock space, while the operator matrix $H(K)$ the Hamiltonian of the system with at most three particles on a lattice.

Main aim of the present paper are

(i) to investigate the spectrum of a family of generalized Friedrichs model;

(ii) to introduce the channel operator $H_{\rm ch}(K)$ corresponding to $H(K)$
and describe its spectrum;

(iii) to obtain an analogue of the Faddeev equation for eigenvectors of $H(K);$

(iv) to show that the essential spectrum of $H(K)$ is coincide with the spectrum of $H_{\rm ch}(K)$;

(v) to prove that the essential spectrum of $H(K)$ consists at most three bounded closed intervals;

(vi) to define the new branches of $\sigma_{\rm ess}(H(K))$.

The next sections are devoted to the discussion of these problems.

\section{Family of generalized Friedrichs models and its spectrum}

To study the spectral properties of the operator $H(K)$ we
introduce a family of bounded self-adjoint operators (generalized Friedrichs
models) $h(K,k),$ $K,k\in {\Bbb T}^{\rm d},$ which acts in the Hilbert space
${\mathcal H}_0 \oplus {\mathcal H}_1$ as $2 \times 2$ operator matrices
\begin{equation}\label{generalized Friedrichs model}
h(K,k):=\left( \begin{array}{cc}
h_{00}(K,k) & h_{01}\\
h_{01}^* & h_{11}(K,k)\\
\end{array}
\right),
\end{equation}
with matrix elements
\begin{align*}
& h_{00}(K,k)f_0=w_1(K,k) f_0,\,\,
h_{01}f_1=\frac{1}{\sqrt{2}} \int_{{\Bbb T}^{\rm d}} v_1(s)
f_1(s)ds,\\
& (h_{11}(K,k)f_1)(q)=w_2(K;k,q)f_1(q), \quad f_i \in {\mathcal H}_i, \quad i=0,1.
\end{align*}
Here
$$
h_{01}^*: {\mathcal H}_0 \to {\mathcal H}_1, \quad (h_{01}^*f_0)(p)=v_1(p)f_0, \quad f_0 \in {\mathcal H}_0.
$$

Now we study some spectral properties of the family $h(K,k),$ given by \eqref{generalized Friedrichs model},
which plays a crucial role in the study of the essential spectrum of $H(K).$
We notice that the spectrum, usual eigenvalues, threshold eigenvalues and virtual levels of the typical models
for a fixed $K,k \in {\Bbb T}^{\rm d}$ are studied in many works, see for example \cite{ALR, ALR1, RD2019}.
The threshold eigenvalues and
virtual levels of a slightly simpler version of $h(K,k)$ were investigated in \cite{RD14}, and the
structure of the numerical range are studied.

Let the operator $h_0(K,k),$ $K,k\in {\Bbb T}^{\rm d}$ acts in
${\mathcal H}_0\oplus {\mathcal H}_1$ as
$$
h_0(K,k):=\left( \begin{array}{cc}
0 & 0\\
0 & h_{11}(K,k)\\
\end{array}
\right).
$$

The perturbation $h(K,k)-h_0(K,k)$ of the operator $h_0(K,k)$ is a self-adjoint operator of rank 2,
and thus, according to the Weyl theorem, the essential spectrum of the operator $h(K,k)$
coincides with the essential spectrum of $h_0(K,k).$ It
is evident that
$$
\sigma_{\rm ess}(h_0(K,k))=[E_{\rm min}(K,k);
E_{\rm max}(K,k)],
$$
where the numbers $E_{\rm min}(K,k)$ and $E_{\rm max}(K,k)$ are defined by
$$
E_{\rm min}(K,k):= \min_{q\in {\Bbb T}^{\rm d}} w_2(K;k,q) \quad \mbox{and}
\quad E_{\rm max}(K,k):= \max_{q\in {\Bbb T}^{\rm d}} w_2(K;k,q).
$$
This yields $\sigma_{\rm ess}(h(K,k))=[E_{\rm min}(K,k); E_{\rm max}(K,k)].$

\begin{rem}
We remark that for some $K, k\in {\Bbb T}^{\rm d}$ the essential spectrum of $h(K,k)$
may degenerate to the set consisting of the unique point $ [E_{\rm min}(K,k); E_{\rm min}(K,k)]$ and hence we can not state that the essential spectrum of $h(K,k)$ is absolutely continuous for any $K, k\in {\Bbb T}^{\rm d}$. For example, this is the case if the function $w_2(\cdot; \cdot,\cdot)$ is of the form
$$
w_2(K; k,q)=\varepsilon(k)+\varepsilon(q)+\varepsilon(K-k-q),
$$
$K=\bar{0}:=(0,\ldots,0),\, k=\bar{\pi}:=(\pi,\ldots,\pi) \in {\Bbb T}^{\rm d}$ and
$$
\varepsilon(q)=\sum_{i=1}^{\rm d}(1-\cos q_i), \quad
q=(q_1,\ldots,q_{\rm d}) \in {\Bbb T}^{\rm d}.
$$
Then $\sigma_{\rm ess}(h(\bar{0},\bar{\pi}))=\{4{\rm d}\}.$
\end{rem}

For any fixed $K, k\in {\Bbb T}^{\rm d}$ we define an analytic function $\Delta_K(k\,; \cdot)$
(the Fredholm determinant associated with the
operator $h(K,k)$) in ${\Bbb C}\setminus [E_{\rm min}(K,k); E_{\rm max}(K,k)]$ by
$$
\Delta_K(k\,; z):=w_1(K; k)-z-\frac{1}{2}\int_{{\Bbb T}^{\rm d}}\frac{v_1^2(t)dt}{w_2(K; k,t)-z}.
$$

The following lemma \cite{ALR} is a simple consequence of the
Birman-Schwinger principle and the Fredholm theorem.

\begin{lem}\label{LEM 1} For any $K,k\in {\Bbb T}^{\rm d}$ the operator $h(K,k)$ has an eigenvalue
$z(K,k) \in {\Bbb C} \setminus [E_{\rm min}(K,k); E_{\rm max}(K,k)]$ if and
only if $\Delta_K(k\,; z(K,k))=0.$
\end{lem}

From Lemma \ref{LEM 1} it follows that for the discrete spectrum of $h(K,k)$ the equality
$$
\sigma_{\rm disc}(h(K,k))=\{z \in {\Bbb C} \setminus [E_{\rm min}(K,k); E_{\rm max}(K,k)]:\,
\Delta_K(k\,; z)=0 \}
$$
holds.

The following lemma describes the number and location of the eigenvalues of
$h(K,k).$

\begin{lem}\label{LEM 2}
For any fixed $K,k\in {\Bbb T}^{\rm d}$ the operator $h(K,k)$ has no
more than one simple eigenvalue lying on the l.h.s. $($resp. r.h.s.$)$ of $E_{\rm min}(K,k)$ $($resp. $E_{\rm max}(K,k))$.
\end{lem}

The proof of Lemma \ref{LEM 2} is an elementary and it follows from the fact that for any fixed $K,k\in {\Bbb T}^{\rm d}$
the function $\Delta_K(k\,; \cdot)$ is a monotonically decreasing on $(-\infty; E_{\rm min}(K,k))$ and $(E_{\rm max}(K,k); +\infty)$.

\section{The spectrum of channel operator corresponding to $H(K)$}

In this section we define the channel operator $H_{\rm ch}(K)$ corresponding to $H(K)$
and describe its spectrum by the spectrum of the family of operators
$h(K,k)$, $K,k \in {\Bbb T}^{\rm d}$, defined by \eqref{generalized Friedrichs model}.

We introduce so-called {\it channel operator} $H_{\rm ch}(K)$ acting
in $L^2({\Bbb T}^{\rm d}) \oplus L^2(({\Bbb T}^{\rm d})^2)$ as a family of $2 \times 2$ operator matrices
\begin{equation}\label{channel operator}
H_{\rm ch}(K):=\left( \begin{array}{cc}
H_{11}(K) & \frac{1}{\sqrt{2}}H_{12}\\
\frac{1}{\sqrt{2}}H_{12}^* & H_{22}(K)\\
\end{array}
\right), \quad K \in {\Bbb T}^{\rm d}.
\end{equation}
It is important that for this case the operator $H_{12}^*$ is defined as follows
$$
H_{12}^*: L^2({\Bbb T}^{\rm d}) \to L^2(({\Bbb T}^{\rm d})^2), \quad (H_{12}^*f_1)(p,q)=v_1(q)f_1(p), \quad f_1 \in L^2({\Bbb T}^{\rm d}).
$$
Under these assumptions the operator $H_{\rm ch}(K)$ is bounded
and self-adjoint.

Set
\begin{align*}
& m_K:= \min\limits_{p,q\in {\Bbb T}^{\rm d}} w_2(K; p,q),
\quad M_K:= \max\limits_{p,q\in {\Bbb T}^{\rm d}} w_2(K; p,q), \\
& \Lambda_K:=\bigcup_{k \in {\Bbb T}^{\rm d}}
\sigma_{\rm disc}(h(K,k)), \quad \Sigma_K:=[m_K; M_K] \cup \Lambda_K.
\end{align*}
Here by Lemma \ref{LEM 1} we may define the set $\Lambda_K$ as the set of all complex numbers
$z \in {\Bbb C} \setminus [E_{\rm min}(K,k); E_{\rm max}(K,k)]$ such that the equality
$\Delta_K(k\,; z)=0$ holds for some $k\in {\Bbb T}^{\rm d}$.

The spectrum of the operator $H_{\rm ch}(K)$ can be precisely
described by the spectrum of the family $h(K,k)$ of
generalized Friedrichs models as well as in the following assertion.

\begin{thm}\label{spectrum of channel operator}
The operator $H_{\rm ch}(K)$ has a purely essential spectrum
and for its spectrum the equality $\sigma(H_{\rm ch}(K))=\Sigma_K$ holds.
\end{thm}

\begin{proof}
It is clear that the operator $H_{\rm ch}(K)$ commutes with any
multiplication operator $U_\alpha$ by the bounded function
$\alpha(\cdot)$ on ${\Bbb T}^{\rm d}$:
\begin{equation*}
U_\alpha\left( \begin{array}{cc}
g_1(p)\\
g_2(p, q)\\
\end{array}
\right)=\left( \begin{array}{cc}
\alpha(p) g_1(p)\\
\alpha(p) g_2(p, q)\\
\end{array}
\right),\,\left( \begin{array}{cc}
g_1\\
g_2\\
\end{array}
\right) \in L^2({\Bbb T}^{\rm d}) \oplus L^2(({\Bbb T}^{\rm d})^2).
\end{equation*}

Therefore the decomposition of the space $L^2({\Bbb T}^{\rm d}) \oplus L^2(({\Bbb T}^{\rm d})^2)$ into
the direct integral
\begin{equation}\label{direct integral}
L^2({\Bbb T}^{\rm d}) \oplus L^2(({\Bbb T}^{\rm d})^2)= \int_{{\Bbb T}^{\rm d}} \oplus \,({\mathcal H}_0 \oplus {\mathcal H}_1)dk
\end{equation}
yields the decomposition into the direct integral
\begin{equation}\label{decom H_3}
H_{\rm ch}(K)= \int_{{\Bbb T}^{\rm d}} \oplus\, h(K,k)dk,
\end{equation}
where the fiber operators (a family of the generalized Friedrichs models) $h(K,k)$ are defined by \eqref{generalized Friedrichs model}.
We note that identical fibers appear in the direct integral in decomposition \eqref{direct integral}.
Then the theorem on the spectrum of decomposable operators \cite{RS4} gives the equality
\begin{equation}\label{spectrum of channel operator 1}
\sigma(H_{\rm ch}(K))=\bigcup_{k \in {\Bbb T}^{\rm d}} \sigma(h(K,k)).
\end{equation}

The definition of the set $\Lambda_K$ and the equality
$$
\bigcup_{k \in {\Bbb T}^{\rm d}} [E_{\rm min}(K,k); E_{\rm max}(K,k)]=[m_K; M_K]
$$
imply the equality
\begin{equation}\label{spectrum of channel operator 2}
\bigcup_{k \in {\Bbb T}^{\rm d}} \sigma(h(K,k))=\Sigma_K.
\end{equation}
Now, the equalities \eqref{spectrum of channel operator 1} and \eqref{spectrum of channel operator 2} complete the proof.
\end{proof}

\section{The Faddeev equation and main property}

In this section we derive an analog of the Faddeev type
system of integral equations for the eigenvectors corresponding
to the discrete eigenvalues (isolated eigenvalues with finite multiplicity) of $H(K)$,
which plays a crucial role in our analysis of the spectrum of $H(K)$.

For any fixed $K\in {\Bbb T}^{\rm d}$ and $z \in {\Bbb C} \setminus \Sigma_K$ we introduce a
$2 \times 2$ block operator matrix $T(K,z)$ acting in ${\mathcal H}_0 \oplus {\mathcal H}_1$ as
$$
T(K,z):=\left( \begin{array}{cc}
T_{00}(K,z) & T_{01}(K,z)\\
T_{10}(K,z) & T_{11}(K,z)
\end{array}
\right),
$$
where the entries $T_{ij}(K,z): {\mathcal H}_j \to {\mathcal H}_i,$ $i,j=0,1$
are defined by
$$
T_{00}(K,z)g_0=(1+w_0(K)-z)g_0,\quad
T_{01}(K,z)g_1=\int_{{\Bbb T}^{\rm d}} v_0(t)g_1(t)dt;
$$
$$
(T_{10}(K,z)g_0)(p)=-\frac{v_0(p)g_0}{\Delta_K(p\,; z)}, \quad
(T_{11}(K,z)g_1)(p)=\frac{v_1(p)}{2 \Delta_K(p\,; z)} \int_{{\Bbb T}^{\rm d}}
\frac{v_1(t)g_1(t)dt}{w_2(K;p,t)-z}.
$$
Here $g_i\in {\mathcal H}_i,$ $i=0,1.$ We note that for each $K \in {\Bbb T}^{\rm d}$ and
$z \in {\Bbb C} \setminus \Sigma_K$ the entries $T_{00}(K,z)$, $T_{01}(K,z)$ and $T_{10}(K,z)$ are rank 1 operators,
the integral operator $T_{11}(K,z)$ belongs to the Hilbert-Schmidt class and
therefore, $T(K,z)$ is a compact operator.

The following theorem is an analog of the well-known Faddeev's
result for the operator $H(K)$ and establishes a connection between
eigenvalues of $H(K)$ and $T(K,z).$

\begin{thm}\label{Main Theorem 1}
The number $z \in {\Bbb C} \setminus \Sigma_K$
is an eigenvalue of the operator $H(K)$ if and only if the number
$\lambda=1$ is an eigenvalue of the operator $T(K,z).$
Moreover the eigenvalues $z$ and $1$ have the same multiplicities.
\end{thm}

\begin{proof}
Let $z \in {\Bbb C} \setminus \sigma_{\rm ess}(H(K))$ be an
eigenvalue of the operator $H(K)$ and $f=(f_0,f_1,f_2) \in {\mathcal
H}$ be the corresponding eigenvector. Then $f_0,$ $f_1$ and $f_2$
satisfy the system of equations
\begin{eqnarray}
&&(w_0(K)-z)f_0+\int_{{\Bbb T}^{\rm d}} v_0(t)f_1(t)dt=0;\nonumber\\
&&v_0(p)f_0+(w_1(K;p)-z)f_1(p)+\int_{{\Bbb T}^{\rm d}} v_1(t) f_2(p,t)dt=0;\label{Fad-Eq1}\\
&&\frac{1}{2}(v_1(p)f_1(q)+v_1(q)f_1(p))+(w_2(K;p,q)-z)f_2(p,q)=0.\nonumber
\end{eqnarray}
The condition $z \not\in [m_K, M_K]$ yields that the inequality $w_2(K;p,q)-z\neq 0$ holds for all $p,q \in {\Bbb T}^{\rm d}.$
Then from the third equation of the system \eqref{Fad-Eq1} for
$f_2$ we have
\begin{equation}\label{Fad-Eq2}
f_2(p,q)=-\frac{v_1(q) f_1(p)+v_1(p) f_1(q)}{2(w_2(K;p,q)-z)}.
\end{equation}

Substituting the expression \eqref{Fad-Eq2} for $f_2$ into the second
equation of the system \eqref{Fad-Eq1}, we obtain that the system of equations
\begin{eqnarray}
&& (w_0(K)-z)f_0+\int_{{\Bbb T}^{\rm d}} v_0(t)f_1(t)dt=0;\nonumber\\
&& v_0(p)f_0+\Delta_K(p\,;z)f_1(p)-\frac{v_1(p)}{2}\int_{{\Bbb T}^{\rm d}} \frac{v_1(t)f_1(t)dt}{w_2(K;p,t)-z}=0
\label{Fad-Eq3}
\end{eqnarray}
has a nontrivial solution and this system of equations has a nontrivial solution if and only if the system of equations
\eqref{Fad-Eq1} has a nontrivial solution.

By the definition of the set $\Lambda_K$ the inequality $\Delta_K(p\,;z) \neq 0$ holds for all $z
\not\in \Lambda_K$ and $p \in {\Bbb T}^{\rm d}.$
Therefore, the system of equations \eqref{Fad-Eq3} has a nontrivial solution if and
only if the following system of equations
\begin{eqnarray*}
&&f_0=(1+w_0(K)-z)f_0+\int_{{\Bbb T}^{\rm d}} v_0(t)f_1(t)dt;\nonumber\\
&&f_1(p)=-\frac{v_0(p)f_0}{\Delta_K(p\,;z)}+\frac{v_1(p)}{2\Delta_K(p\,;z)}\int_{{\Bbb
T}^{\rm d}} \frac{v_1(t)f_1(t)dt}{w_2(K;p,t)-z}
\end{eqnarray*}
or $2 \times 2$ matrix equation
\begin{equation}\label{Fad-Eq4}
g=T(z)g,\quad g=(f_0,f_1) \in {\mathcal H}_0 \oplus {\mathcal H}_1
\end{equation}
has a nontrivial solution and the linear subspaces of solutions
of (\ref{Fad-Eq1}) and (\ref{Fad-Eq4}) have the same dimension.
\end{proof}

\begin{rem}
We point out that the matrix equation \ref{Fad-Eq4}
is an analogue of the Faddeev type system of integral equations for eigenfunctions of the
operator $H(K)$ and it is played crucial role in the analysis of the
spectrum of $H(K).$
\end{rem}

\section{Essential spectrum and its new branches}

In this section applying the statements of Sections 3-5, the Weyl criterion \cite{RS4}
we investigate the essential spectrum of $H(K)$.

Denote by $\|\cdot\|$ and $(\cdot,\cdot)$ the norm and scalar product in the corresponding
Hilbert spaces.

For the convenience of the reader we formulate Weyl's criterion for the essential spectrum of $H(K)$ as follows. First, a number $\lambda$ is in the spectrum of $H(K)$ if and only if there exists a sequence $\{F_n(K) \}$ in the space ${\mathcal H}$ such that $||F_n(K)||=1$ and
\begin{equation}\label{orthonormal system}
\lim\limits_{n \to \infty} \Vert (H(K)-z_0(K) E) F_n(K) \Vert = 0.
\end{equation}
Here $E$ is an identity operator on ${\mathcal H}.$
Furthermore, $\lambda$ is in the essential spectrum if there is a sequence satisfying this condition, but such that it contains no convergent subsequence (this is the case if, for example $\{F_n(K) \}$ is an orthonormal sequence); such a sequence is called a singular sequence.

The following theorem describes the location of
the essential spectrum of $H(K)$.

\begin{thm}\label{Main Theorem 2} The essential spectrum of $H(K)$ is coincide with the spectrum of $H_{\rm ch}(K)$, that is,
$\sigma_{\rm ess}(H(K))=\sigma(H_{\rm ch}(K))$.
Moreover the set $\sigma_{\rm ess}(H(K))$ consists no more than three
bounded closed intervals.
\end{thm}

\begin{proof}
By the Theorem \ref{spectrum of channel operator}, we have $\sigma(H_{\rm ch}(K))=\Sigma_K.$
Therefore, we must to show that $\sigma_{\rm ess}(H(K))=\Sigma_K.$
We begin by proving $\Sigma_K \subset \sigma_{\rm ess}(H(K)).$ Since
the set $\Sigma_K$ has form
$\Sigma_K=\Lambda_K \cup [m_K; M_K]$ first we show that $[m_K; M_K] \subset \sigma_{\rm ess}(H(K)).$
Let $z_0(K) \in [m_K; M_K]$ be an arbitrary point. We prove that $z_0(K) \in
\sigma_{\rm ess}(H(K)).$ To this end it is suffices to construct a sequence of orthonormal vector-functions $\{F_n(K) \} \subset
{\mathcal H}$ satisfying \eqref{orthonormal system}.

From continuity of the function $w_2(K;\cdot,\cdot)$ on the compact set $({\Bbb T}^{\rm
d})^2$ it follows that there exists some point $(p_0(K),q_0(K)) \in ({\Bbb T}^{\rm d})^2$ such that $z_0(K)=w_2(K; p_0(K),q_0(K)).$

For $n \in {\Bbb N}$ we consider the following neighborhood of the point $(p_0(K),q_0(K)) \in ({\Bbb
T}^{\rm d})^2:$
$$
W_n(K):=V_n(p_0(K)) \times V_n(q_0(K)),
$$
where
$$
V_n(p_0(K)):=\Bigl\{p \in {\Bbb T}^{\rm d}:
\frac{1}{n+n'+1}<|p-p_0(K)|<\frac{1}{n+n'} \Bigr\}
$$
is the punctured neighborhood of the point $p_0(K) \in {\Bbb T}^{\rm d}$ and $n' \in {\Bbb N}$ is chosen in such way that
$V_n(p_0(K)) \cap V_n(q_0(K))=\emptyset$ for all $n \in {\Bbb N}$ (provided that $p_0(K) \neq q_0(K)$).

Let $\mu(\Omega)$ be the Lebesgue measure of the set $\Omega$ and $\chi_\Omega(\cdot)$
be the characteristic function of the set $\Omega.$
We choose the sequence of functions $\{F_n(K)\} \subset {\mathcal H}$ as follows:
\begin{align*}
F_n(K):=\frac{1}{\sqrt{2\mu(W_n(K))}}\left( \begin{array}{ccc}
0\\
0\\
\chi_{W_n(K)}(p,q)+\chi_{W_n(K)}(q,p)
\end{array}
\right).
\end{align*}

It is clear that $\{F_n(K)\}$ is an orthonormal sequence.

For any $n \in {\Bbb N}$ let us consider an element $(H(K)-z_0(K) E)F_n(K)$ and estimate its norm:
$$
\|(H(K)-z_0(K) E)F_n(K)\|^2 \le \sup\limits_{(p,q) \in W_n(K)}
|w_2(K;p,q)-z_0(K)|^2+\, \mu(V_n(p_0(K)))
\max\limits_{p \in {\Bbb T}^{\rm d}} |v_1(p)|^2.
$$

The construction of the set $V_n(p_0(K))$ and the continuity of the function
$w_2(K;\cdot,\cdot)$ implies $\|(H(K)-z_0(K) E)F_n(K)\| \to 0$ as $n
\to \infty,$ i.e. $z_0(K) \in \sigma_{\rm ess}(H(K)).$ Since the point $z_0(K)$ is an arbitrary
we have $[m_K; M_K] \subset \sigma_{\rm ess}(H(K)).$

Now let us prove that $\Lambda_K \subset \sigma_{\rm ess}(H(K)).$ Taking an
arbitrary point $z_1(K) \in \Lambda_K$ we show that $z_1(K) \in \sigma_{\rm ess}(H(K)).$
Two cases are possible: $z_1(K) \in [m_K; M_K]$ or $z_1(K) \not\in [m_K; M_K].$
If $z_1(K) \in [m_K; M_K],$ then it is already proven above that $z_1(K) \in \sigma_{\rm ess}(H(K)).$
Let $z_1(K) \in \Lambda_K \setminus [m_K; M_K].$
Definition of the set $\Lambda_K$ and Lemma \ref{LEM 1} imply that
there exists a point $p_1(K) \in {\Bbb T}^{\rm d}$ such that $\Delta_K(p_1(K)\,; z_1(K))=0.$

We choose a sequence of orthogonal vector-functions $\{\Phi_n(K)\}$ as
\begin{align*}
\Phi_n(K):=\left( \begin{array}{ccc}
0\\
\phi_1^{(n)}(K; p)\\
\phi_2^{(n)}(K; p,q)
\end{array}
\right), \quad \mbox{where} \quad & \phi_1^{(n)}(K; p):=\frac{c_n(K; p) \chi_{V_n(p_1(K))}(p)}{\sqrt{\mu(V_n(p_1(K)))}},\\
& \phi_2^{(n)}(K; p,q):=-\frac{v_1(p) \phi_1^{(n)}(K; q)+v_1(q) \phi_1^{(n)}(K; p)}{2(w_2(K; p,q)-z_1(K))}.
\end{align*}
Here $c_n(K; p) \in L_2({\Bbb T}^{\rm d})$ is chosen from the orthonormality condition for $\{\Phi_n(K)\},$
that is, from the condition
\begin{align}\label{orthonormality}
(\Phi_n(K), \Phi_m(K)) &= \frac{1}{2\sqrt{\mu (V_n(p_1(K)))}
\sqrt{\mu (V_m(p_1(K)))}} \\
& \times \int_{V_n(p_1(K))} \int_{V_m(p_1(K))} \frac{v_1(p)v_1(q) c_n(K; p)c_m(K; q)}{(w_2(K;p,q)-z_1(K))^2} dpdq=0\nonumber
\end{align}
for $n \neq m.$
The existence of $\{c_n(K; \cdot)\}$ is a consequence of the following proposition.

\begin{prop}\label{PROP1}
There exists an orthonormal system $\{c_n(K; \cdot)\} \subset L_2({\Bbb T}^{\rm d})$ satisfying the conditions
${\rm supp}\,c_n(K; \cdot) \subset V_n(p_1(K))$ and $(\ref{orthonormality}).$
\end{prop}

\begin{proof}[Proof of Proposition $\ref{PROP1}$]
We construct the sequence $\{c_n(K; \cdot)\}$ by the induction method.
Suppose that $c_1(K; p):=\chi_{V_1(p_1(K))}(p) \left(\sqrt{\mu(V_1(p_1(K)))}\right)^{-1}.$
Now we choose the function $\widetilde{c}_2(K; \cdot) \in L_2(V_2(p_1(K)))$ so that $\Vert \widetilde{c}_2(K; \cdot) \Vert =1$
and $(\widetilde{c}_2(K; \cdot), \varepsilon_1^{(2)}(K; \cdot))=0,$ where
$$
\varepsilon_1^{(2)}(K; p):= v_1(p)\chi_{V_2(p_1(K))}(p) \int_{{\Bbb T}^{\rm d}}
\frac{v_1(q)c_1(K; q)dq}{(w_2(K; p,q)-z_1(K))^2}.
$$

Set $c_2(K; p):=\widetilde{c}_2(K; p) \chi_{V_1(p_1(K))}(p).$ We continue this process.
Suppose that $c_1(K; p),$ $\ldots,$ $c_n(K; p)$ are constructed. Then the function
$\widetilde{c}_{n+1}(K; \cdot) \in L_2(V_{n+1}(p_1(K)))$ is chosen so that it is
orthogonal to all functions
$$
\varepsilon_m^{(n+1)}(K; p):= v_1(p) \chi_{V_{n+1}(p_1(K))}(p) \int_{{\Bbb T}^{\rm d}}
\frac{v_1(q)c_m(K; q)dq}{(w_2(K; p,q)-z_1(K))^2}, \quad
m=1,\ldots,n
$$
and $\|\widetilde{c}_{n+1}(K; \cdot) \|=1.$ Let $c_{n+1}(K; p):= \widetilde{c}_{n+1}(K; p) \chi_{V_{n+1}(p_1(K))}(p).$
Thus, we have constructed the orthonormal system of functions $\{c_n(K; \cdot)\}$ satisfying the assumptions
of the proposition. Proposition \ref{PROP1} is proved.
\end{proof}

We continue the proof of Theorem \ref{Main Theorem 2}. Now for $n \in {\Bbb N}$ we consider $(H(K)-z_1(K) E)\Phi_n(K)$ and estimate its norm as
\begin{align}
\|(H(K)-z_1(K) E)\Phi_n(K)\|^2 & \leq C(K) \max\limits_{p \in {\Bbb T}^{\rm d}} v_1^2(p)\, \mu(V_{n}(p_1(K)))\nonumber \\
& + 2 \sup\limits_{p \in V_n(p_1(K))} |\Delta_K(p;\, z_1(K))|^2 \label{norm estimate}
\end{align}
for some $C(K)>0.$

Since $\mu(V_{n}(p_1(K))) \to 0$ and $\sup\limits_{p \in V_n(p_1(K))} |\Delta_K(p;\, z_1(K))|^2 \to 0$ as $n \to \infty$,
from the \eqref{norm estimate} it follows that
$\|(H(K)-z_1(K) E)\Phi_n(K)\| \to 0$ as $n \to \infty.$ This implies $z_1(K) \in \sigma_{\rm ess}(H(K)).$
Since the point $z_1(K)$ is an arbitrary, we have $\Lambda_K \subset \sigma_{\rm ess}(H(K)).$
Therefore, we proved that $\Sigma_K \subset \sigma_{\rm ess}(H).$

Now we prove the inverse inclusion, i.e. $\sigma_{\rm ess}(H(K)) \subset
\Sigma_K.$ Since for each $z\in {\Bbb C} \setminus \Sigma_K$
the operator $T(K; z)$ is a compact-operator-valued
function on ${\Bbb C} \backslash \Sigma_K,$ from the
self-adjointness of $H(K)$ and Theorem \ref{Main Theorem 1} it follows that
the operator $(I-T(K, z))^{-1}$ exists if $z$ is real and has a large absolute value.
The analytic Fredholm theorem (see, e.g., Theorem VI.14 in \cite{RS4}) implies that
there is a discrete set $S_K \subset {\Bbb C} \setminus \Sigma_K$ such that the function
$(I-T(K, z))^{-1}$ exists and is analytic on ${\Bbb C} \setminus (S_K \cup \Sigma_K)$
and is meromorphic on ${\Bbb C} \setminus \Sigma_K$ with finite-rank residues.
This implies that the set $\sigma(H(K)) \setminus \Sigma_K$
consists of isolated points, and the only possible accumulation points of $\Sigma_K$
can be on the boundary. Thus
$$
\sigma(H(K)) \setminus \Sigma_K \subset
\sigma_{\rm disc}(H(K))=\sigma(H(K)) \setminus \sigma_{\rm ess}(H(K)).
$$
Therefore,
the inclusion $\sigma_{\rm ess}(H(K)) \subset \Sigma_K$ holds.
Finally we obtain the equality $\sigma_{\rm ess}(H(K)) = \Sigma_K.$

By Lemma \ref{LEM 2} for any $K,k \in {\Bbb T}^{\rm d}$ the operator
$h(K,k)$ has no more than one simple eigenvalue on the l.h.s. (resp. r.h.s) of $E_{\rm min}(K,k)$ (resp. $E_{\rm max}(K,k)$).
Then by the theorem on the spectrum
of decomposable operators \cite{RS4} and by the definition of the set
$\Lambda_K$ it follows that the set $\Lambda_K$ consists of the union of no more than two bounded closed intervals.
Therefore, the set $\Sigma_K$ consists of the union of no more than three
bounded closed intervals. Theorem \ref{Main Theorem 2} is completely proved.
\end{proof}

In the following we introduce the new subsets of the essential spectrum of $H(K).$

\begin{defn}
The sets $\Lambda_K$ and $[m_K; M_K]$ are called two- and three-particle
branches of the essential spectrum of $H(K),$ respectively.
\end{defn}

It is obvious that for a given $H(K)$, the operator $H_{\rm ch}(K)$ is uniquely selected by the property of decomposability into a direct integral.

According to Theorem \ref{Main Theorem 2}, the operator $H_{\rm ch}(K)$ have the characteristic property of a channel operator of the corresponding discrete Schr\"{o}dinger operator, see \cite{ALM07, LM03}. Therefore, we call this operator the channel operator
associated with $H(K)$.
We note that the channel operator $H_{\rm ch}(K)$ have a more simple structure than the operator $H(K)$, and hence,
Theorem \ref{Main Theorem 2} plays an important role in the subsequent investigations of the essential spectrum of $H(K)$.

Since for any $K \in {\Bbb T}^{\rm d}$ and $z \in {\Bbb C} \setminus \Sigma_K$ the operators $T_{00}(K, z),$
$T_{01}(K, z)$ and $T_{10}(K, z)$ are one dimensional and the kernel of $T_{11}(K, z)$ is a continuous function on $({\Bbb T}^{\rm
d})^2,$ the Fredholm determinant $\Omega_K(z)$ of the operator
$I-T(K, z),$ where $I$ is the identity operator in ${\mathcal H}_0 \oplus {\mathcal H}_1,$ exists and is a real-analytic function on ${\Bbb C}
\setminus \Sigma_K.$

According to Fredholm's theorem \cite{RS4} and Theorem \ref{Main Theorem 1} the number $z \in {\Bbb C} \setminus \Sigma_K$
is an eigenvalue of $H(K)$ if and only if $\Omega_K(z)=0,$ that is,
$$
\sigma_{\rm disc}(H(K))=\{ z \in {\Bbb C} \setminus \Sigma_K:
\Omega_K(z)=0 \}.
$$

\end{document}